\documentclass[a4paper,10pt]{article}
\usepackage[utf8]{inputenc}
\usepackage{graphicx,subfigure}
\usepackage[sort,nocompress]{cite}

\title{Hardness of Liar's Domination on Unit Disk Graphs}
\author{Ramesh K. Jallu and Gautam K. Das\thanks{corresponding author} \\
\small Department of Mathematics\\\small Indian Institute of Technology Guwahati\\ \small\{j.ramesh, gkd\}@iitg.ernet.in}
\usepackage{amsthm}
\usepackage{amsmath}
\newtheorem{theorem}{Theorem}
\newtheorem{lemma}{Lemma}
\begin{document}

 \maketitle

\begin{abstract}
 A unit disk graph is the intersection graph of a set of unit diameter disks in the plane.
 In this paper we consider liar's domination problem on unit disk graphs,
 a variant of dominating set problem. We call this problem as {\it Euclidean
 liar's domination problem}. In the Euclidean liar's domination
 problem, a set ${\cal P}=\{p_1,p_2,\ldots,p_n\}$  of $n$ points
 (disk centers) are given in the Euclidean plane. For $p \in {\cal P}$,
 $N[p]$ is a subset of ${\cal P}$ such that for any $q \in N[p]$, the
 Euclidean distance between $p$ and $q$ is less than or equal to 1, i.e.,
 the corresponding unit diameter disks intersect. The objective of the Euclidean liar's domination problem is to find
 a subset $D\; (\subseteq {\cal P})$ of minimum size having the following
 properties : (i) $|N[p_i] \cap D| \geq 2$ for
 $1 \leq i \leq n$, and (ii) $|(N[p_i] \cup N[p_j]) \cap D| \geq 3$
 for $i\neq j, 1\leq i,j \leq n$. This article aims to prove the Euclidean
 liar's domination problem is NP-complete.
\end{abstract}

\section{Introduction}
Let $G=(V,E)$ be a graph. For a vertex $v \in V$, we define
$N[v] = \{u\in V \mid (v,u)\in E\} \cup \{v\}$. A subset $D$ of $V$ is a {\bf liar's
dominating set} if (i) for every $v \in V$, $|N[v]\cap D| \geq 2$,
and (ii) for every distinct pair of vertices $u$ and $v$,
$|(N[u]\cup N[v])\cap D| \geq 3$. Liar's domination
problem in a graph $G=(V,E)$ asks to find a
liar's dominating set of $G$ with minimum size.

\subsection{Related work}
The liar's domination problem is introduced by Slater in 2009 and
showed that the problem is NP-hard for general graphs \cite{slater}.
Later, Roden and Slater showed that the problem is NP-hard even for
bipartite graphs \cite{roden}. Panda and Paul \cite{panda2013liar}
proved that the problem is NP-hard for split graphs and chordal graphs.
The authors also proposed a linear time algorithm for computing a minimum
cardinality liar's dominating set in a tree.

\subsection{Our work}
A {\it unit disk graph} (UDG) is an intersection graph of a
family of unit diameter disks in the plane.
Given a set $C = \{C_1, C_2,\ldots, C_n\}$ of $n$ circular disks
in the plane, each having diameter 1, the corresponding UDG
$G = (V,E)$ is defined as follows: each vertex $v_i \in V$
corresponds to a disk $C_i \in C$, and there is an edge between
two vertices $v_i$ and $v_j$ if and only if $C_i$ and $C_j$ intersect.

In this paper we consider the geometric version of the liar's domination
problem and we call it as \emph{Euclidean liar's domination problem}.
In the Euclidean liar's domination problem we are given a UDG and a set
${\cal P}$ of $n$ disk centers of the given UDG in the plane. For $p \in {\cal P}$, $N[p]$
is a subset of ${\cal P}$ such that for any $q \in N[p]$, the Euclidean
distance between $p$ and $q$ is less than or equal to 1. We define
$\Delta = \max\{|N[p]| :  p\in{\cal P}\}$. The objective of
the Euclidean liar's domination problem is to find a minimum
size subset $D$ of ${\cal P}$ such that (i) for every point in ${\cal P}$ there exists
at least two points in $D$ which are at most distance one, and
(ii) for every distinct pair of points $p_i$ and $p_j$ in
${\cal P}$, $|(N[p_i]\cup N[p_j])\cap D|\geq 3$, in other words,
the number of points in $D$ that are within unit distance with
points in the closed neighborhood union of $p_i$ and $p_j$ is
at least three.

\section{Complexity}
 In this section we show that the Euclidean liar's domination problem is
 NP-complete for UDGs. The decision version of liar's dominating set
 of a UDG can be defined as follows.

\begin{description}
 \item[UDG LIAR'S DOMINATING SET (UDG-LR-DOM)]
 \item[Instance :] A unit disk graph $G=(V,E)$ and a positive integer $k$.
 \item[Question :] Does there exist a liar's dominating set $L$ of
 $G$ such that $|L|\leq k$?
\end{description}
 We prove the NP-completeness of UDG-LR-DOM by reducing dominating set
 problem defined on a planar graph with maximum degree 3 to it, which
 is known to be NP-complete \cite{gary}. The decision version of  dominating set
 of a planar graph with maximum degree 3 can be defined as follows.

\begin{description}
 \item[PLANAR DOMINATING SET (PLA-DOM)]
 \item[Instance :] A planar graph $G=(V,E)$ with maximum degree 3 and
 a positive integer $k$.
 \item[Question :] Does there exist a dominating set $D$ of
 $G$ such that $|D|\leq k$?
\end{description}

\begin{lemma}[\cite{valiant}] \label{key-lemma}
 A planar graph $G=(V,E)$ with maximum degree 4 can be embedded in the plane
 using $O(|V|)$ area in such a way that its vertices are at integer
 co-ordinates and its edges are drawn so that they are made up of line
 segments of the form $x=i$ or $y=j$, for integers $i$ and $j$.
\end{lemma}
 Algorithms to produce such embeddings are discussed in \cite{itai,hopcroft}.
 Many standard graph theoretic problems on UDGs are shown to be NP-complete
 with the aid of Lemma \ref{key-lemma} \cite{clark}.

\begin{lemma}
Let $G=(V,E)$ be a planar graph with maximum degree 3 and $|E| > 2$. $G$ can be
embedded in the plane such that its vertices are at $(4i,4j)$
and its edges are drawn as a sequences of
consecutive line segments drawn on the lines $x=4i$ or $y=4j$ for
some integers $i$ and $j$ .
\end{lemma}

 In summary, we can draw a planar graph $G=(V,E)$ of maximum degree
 3 on a grid in the plane, where each grid cell is of size $4 \times 4$, such that :
  
 \begin{enumerate}
 \begin{figure}
  \centering
  \mbox{
  \subfigure[]{\includegraphics[scale=0.5]{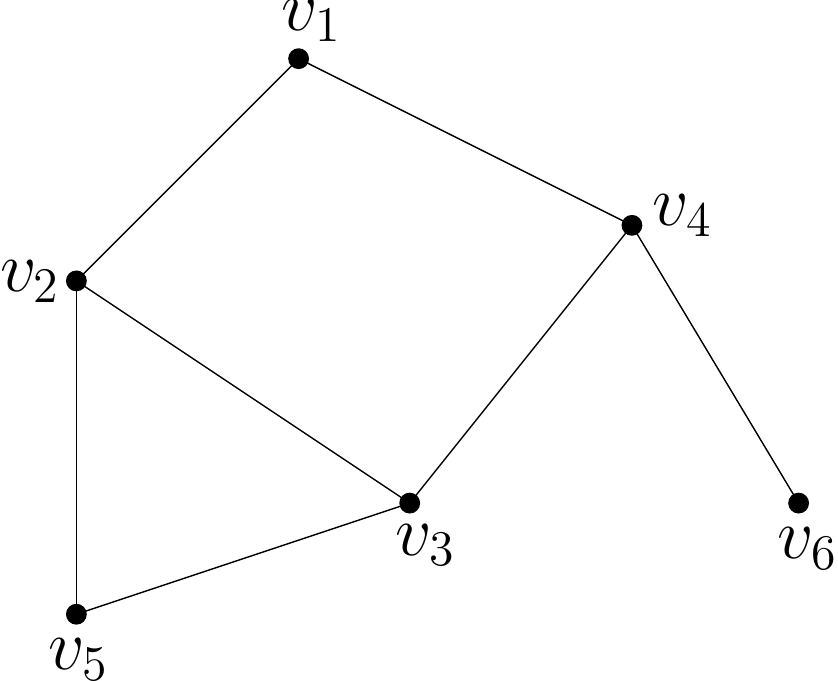}}
  \quad
  \subfigure[]{\includegraphics[scale=0.5]{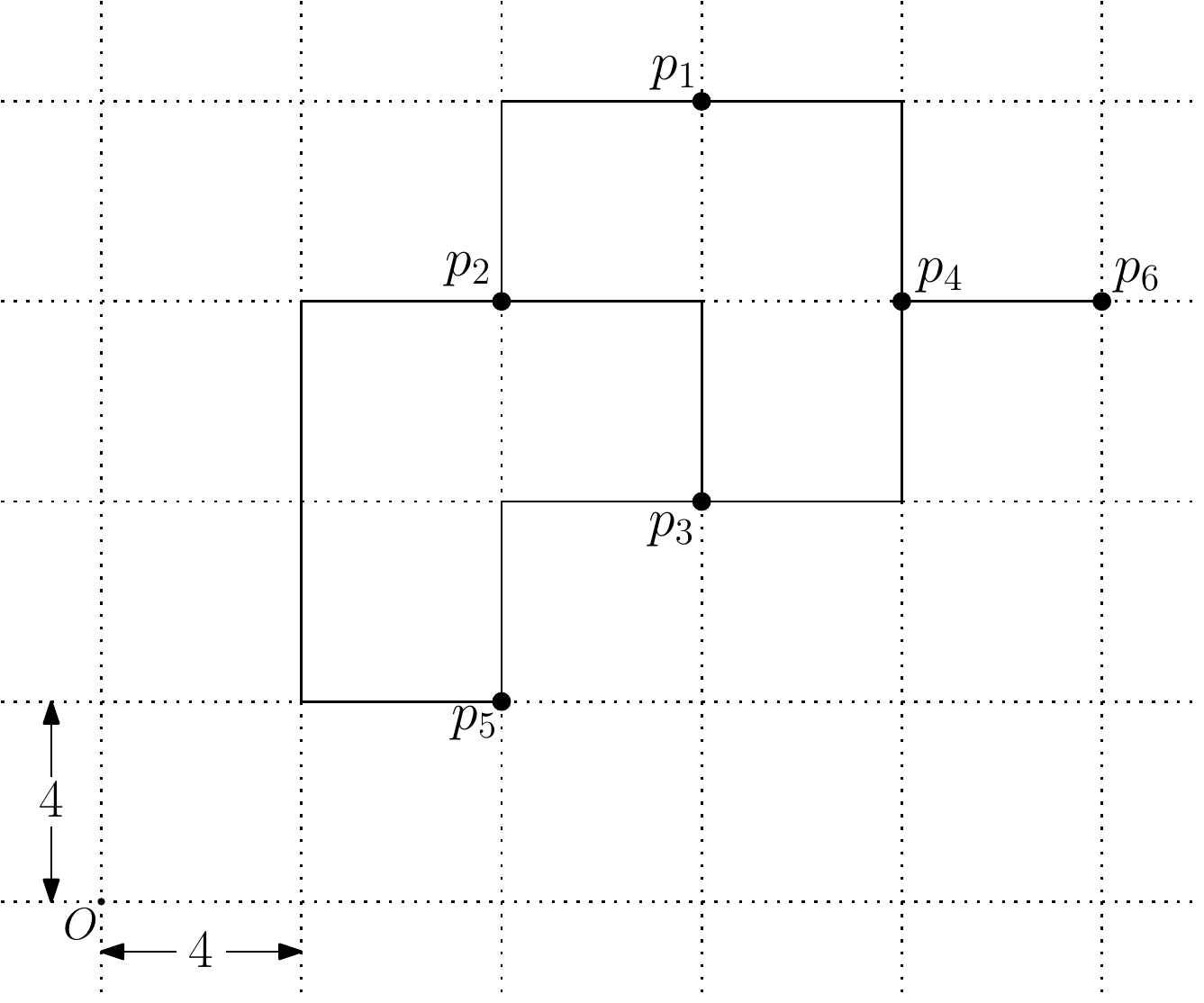}}}
  \caption{
  (a) A planar graph $G$ with maximum degree 3,
  (b) An embedding of $G$ on a grid in the plane.}
  \end{figure}\label{graph_grid}
  \item Each vertex $v_i$ in $G$ is replaced by a point $p_i$ in the plane.
  \item The co-ordinates of each point $p_i$ (corresponding to a vertex $v_i$) are
  $(4i,4j)$ for some integers $i$ and $j$ (see Figure \ref{graph_grid}).
  \item An edge between two points is represented as a sequences of
  consecutive line segments and is drawn on the lines $x=4i$ or $y=4j$
  for some integers $i$ or $j$ (these consecutive line segments may
  bend at some positions of the form $(4i',4j')$).
  \item No two lines representing edges of $G$ intersect each other,
  i.e., any two set of consecutive line segments correspond to two
  distinct edges of $G$ can not have a common point unless the edges
  incident at a vertex in $G$.

 \end{enumerate}
\begin{lemma}\label{lemma-udg}
 A unit disk graph $G'=(V',E')$ can be constructed from the embedding in polynomial time.
\end{lemma}

\begin{proof}
 Let us first embed the graph $G$ in the plane and divide the set of
 line segments in the embedding into two
 categories, namely, proper and improper. We call a line segment is proper if none of its end
 points corresponds to a vertex in $G$. For each edge $(p_i,p_j)$ of
 length 4 units we add four points such that two points at distances 1 and 1.5 units from
 $p_i$ and $p_j$ respectively (see edge $(p_4,p_6)$ in Figure \ref{grid_udg}(a)).
 For each edge of length greater than 4 units,
 we add the following points : for an improper line segment four points
 at distances 1, 1.5, 2.5, and 3.5 units respectively from the end point corresponds to a
 vertex in $G$ and four points for a proper line segment at distances
 0.5 and 1.5 units from its end points (see Figure \ref{grid_udg}(a)).
 If the total number of line segments
 used in the embedding is $l$, then the sum of the lengths of the line
 segments is $4l$ as each line segment has length 4 units.

 \begin{figure}
  \centering
  \mbox{
  \subfigure[]{\includegraphics[scale=0.5]{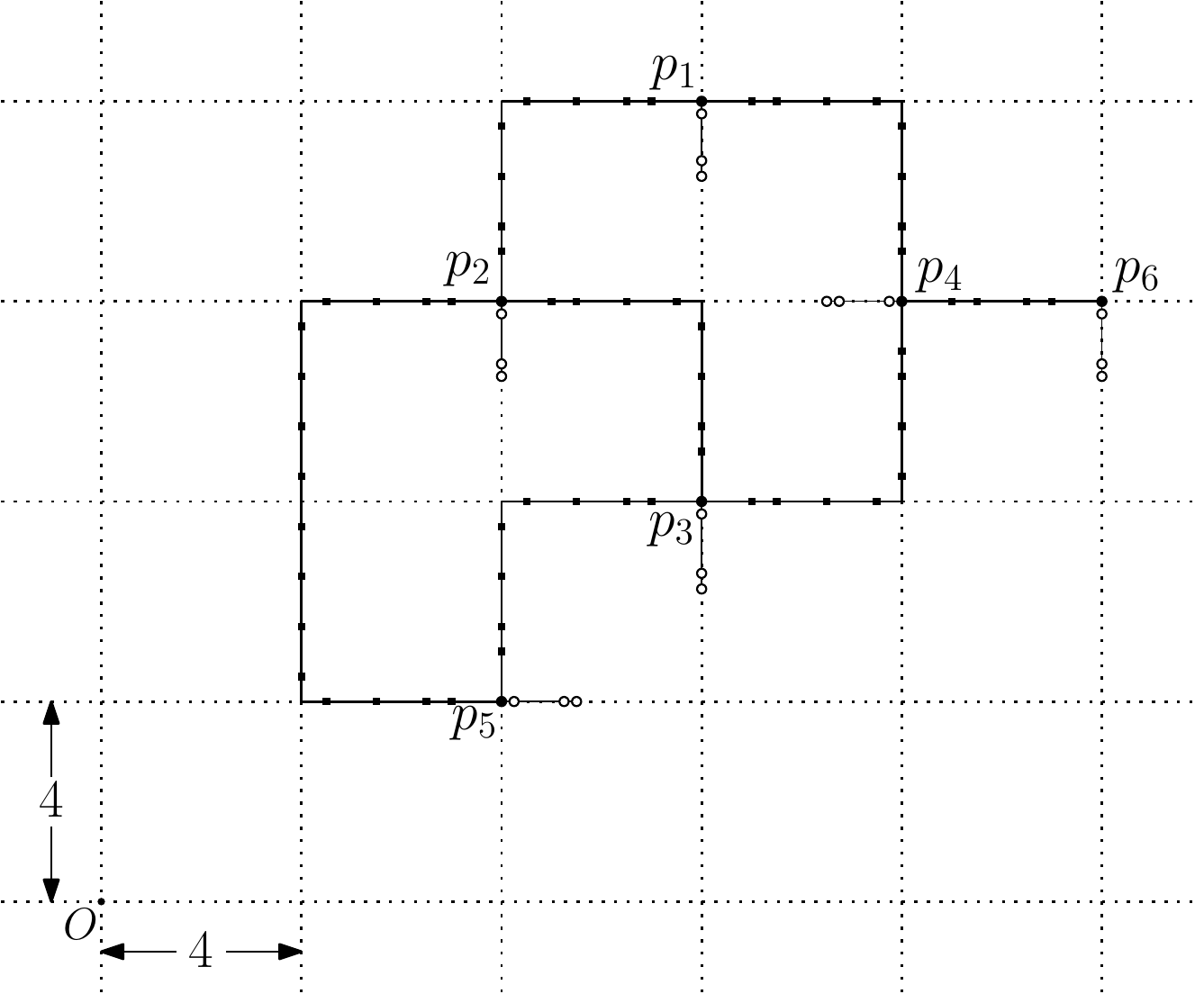}}
  \quad
  \subfigure[]{\includegraphics[scale=0.5]{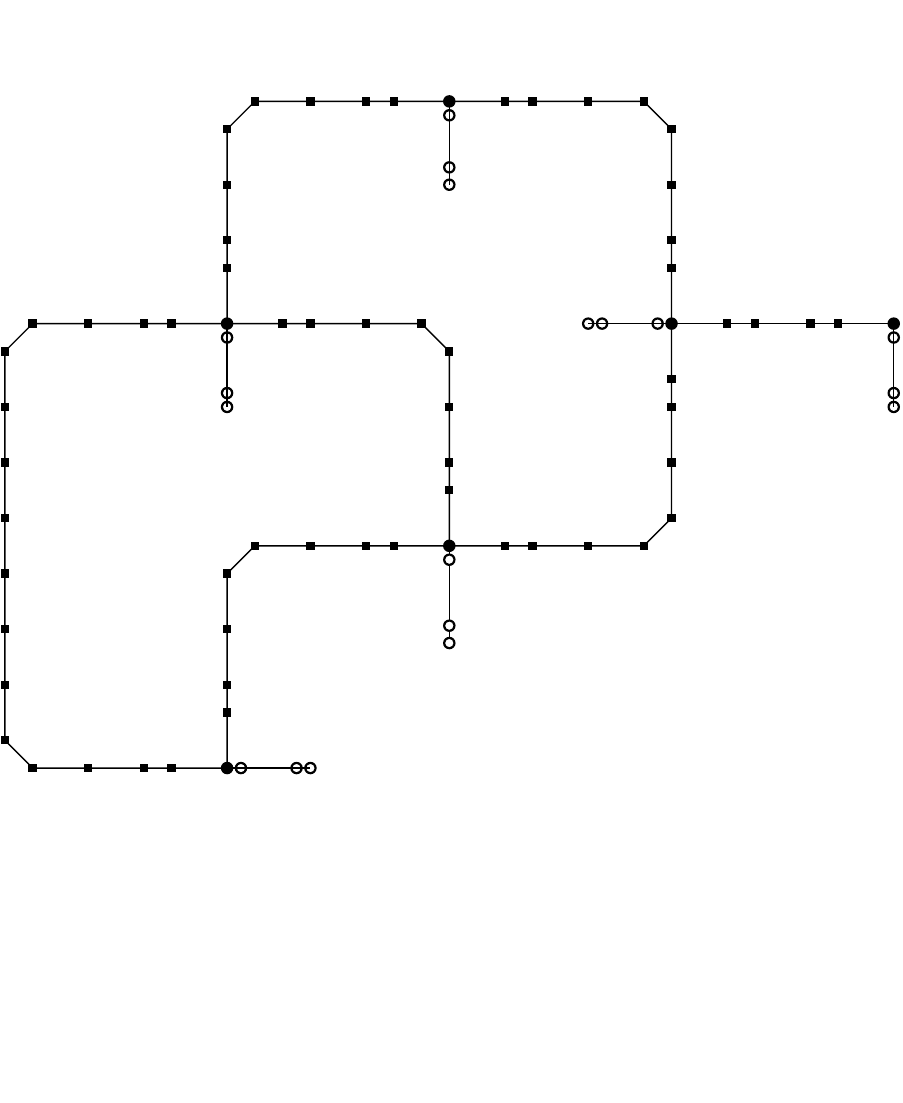}}}
  \caption{
  (a) Construction of unit disk graph from the embedding, and 
  (b) its corresponding unit disk graph.}\label{grid_udg}
  \end{figure}

 Draw a line segment of length 1.4 units (on the lines $x=4i$ or $y=4j$
 for some integers $i$ or $j$) for every point $p_i$ (as shown in
 Figure \ref{grid_udg}(a)) corresponds a vertex $v_i$ in
 $G$ without coinciding with the line segments that had already been drawn before.
 Observe that adding this line segment on the lines $x=4i$ or $y=4j$
 is possible with out loosing the
 planarity as the maximum degree of $G$ is 3. Now, add three points
 (say $x_i$, $y_i$, and $z_i$) at distances 0.2, 1.2, and 1.4 units respectively
 from $p_i$.
 
 For convenience we name the points added (i) correspond to vertices of
 $G$ by {\it node points} (ii) on the line segments of length greater
 than or equal to 4 by {\it joint points}, and (iii) on the line
 segments of length 1.4 by {\it support points}. Let us denote these three
 sets of points by $N$, $J$, and $S$ respectively. In Figure \ref{grid_udg}(a)
 these sets of points respresented as set of solid circles, solid squares,
 and circles respectively. Let $N = \{p_1,p_2,\ldots,p_n\}$,
 $J=\{q_1,q_2,\ldots,q_m\}$, and $S=\{x_i,y_i,z_i \mid 1\leq i\leq n\}$.
 After defining the above sets, remove all the line segments.

 Now we construct a UDG $G'=(V',E')$,
 where $V' = N \cup J \cup S$ and there is an edge between two points in $V'$
 if and only if the Euclidean distance between the points is at most 1
 (see Figure \ref{grid_udg}(b)). Observe
 that, $|N| = n$, $|J|=4l(=m)$, where $l$ is the total length of the segments
 having length greater than or equal to 4, and $|S|=3n$. Hence, $|V'| = 4(n+l)$
 and $l$ is bounded by a polynomial of $n$. Therefore $G'$ can be constructed in
 polynomial time.
 \end{proof}
 
\begin{theorem}
 UDG-LR-DOM is NP-complete.
\end{theorem}
 \begin{proof}
 For any given subset $L$ of $V'$ and a positive integer $k'$, it is easy
 to verify that the subset $L$ is a liar's dominating set of size at most
 $k$ or not. Hence UDG-LR-DOM belongs to the class NP.
 
 We prove the hardness of UDG-LR-DOM by reducing PLA-DOM to it. Let an
 instance, $G=(V,E)$, of PLA-DOM has been given. Construct an instance, a UDG $G'=(V',E')$,
 of UDG-LR-DOM as discussed in Lemma \ref{lemma-udg}.
 We now prove the following claim : {\it $G$ has a dominating set of size at most $k$
 if and only if $G'$ has a liar's dominating set of size at most $k'=k + 4l + 3n$}.

 \noindent {\bf Necessity :} Let $D \subseteq V$ be the given dominating set of $G$ with $|D| \leq k$.
 Let $L = D \cup J \cup S$. We prove that $L$ is a liar's dominating set of $G'$.
 
 (i) Every point $p_i$ in $N$ is dominated once by a point $x_i$ in $S$ and by at
 least one point in $J$. Since $J\subseteq L$, every point in $J$ is
 dominated by itself and by it's neighbor and maybe by one point in $D$.
 Similarly, $S \subseteq L$, every point in $S$ is double dominated by
 points in $L$. Thus, every point in $V'$ satisfies the first condition
 of liar's dominating set.
 
 (ii) Now consider every distinct pair of points in $V'$. Every point $p_i$ in $N$ is
 dominated by $x_i$ and some $q_i$ in $J$. Therefore, $|(N[p_i]\cup N[p_j])\cap L|\geq |\{x_i,q_i,x_j,q_j\}|=4$. Similarly, $|(N[p_i]\cup N[x_j])\cap L|\geq |\{x_i,q_i,x_j,y_j\}|=4$ and
 $|(N[q_i]\cup N[p_j])\cap L|\geq |\{q_i,x_j,q_j\}|=3$. Also $|(N[x_i]\cup N[z_j])\cap L|\geq |\{x_i,y_i,z_j,y_j\}|=4$. In the same way we can prove that
 the rest of the pair combinations have at least three points of
 $L$ in their closed neighborhood union. Thus every distinct pair of points
 in $V'$ satisfies the second condition of liar's dominating set.

 So $L$ is a liar's dominating set of $G'$ and $|L| = |D| + |J| + |S|\leq k+4l+3n=k'$.
 Thus the necessity follows.
 
 \noindent {\bf Sufficiency :} Let $L \subseteq V'$ be a liar's
 dominating set of size at most $k'=k+4l+3n$. We prove that $G$ has a
 dominating set of size at most $k$.

 Observe that we added points $x_i,y_i, z_i$ in such a way that $p_i$
 is adjacent to $x_i$, $x_i$ is adjacent to $y_i$ and $y_i$ is adjacent
 to $z_i$ i.e., $\{(p_i,x_i),(x_i,y_i),(y_i,z_i)\}\subset E'$ for each $i$. Hence,
 $z_i$ and $y_i$ must be in $L$ due to the first condition of liar's domination.
 Also, every component of $\langle L \rangle$ must contain at least three vertices due to the second
 condition of liar's domination. Hence, $x_i \in L$. Therefore, any liar's dominating
 set of $G'$ must contain $\{x_i,y_i,z_i\}, 1\leq i \leq n$ i.e., $S \subset L$.
 These account for $3n$ vertices of $L$. Let $L'=L\setminus S$.
 Now we shall show that, by removing or replacing some points in $L'$, $k$ node points
 can be chosen such that the corresponding vertices in $G$ is a dominating set of $G$.
 Note that $L'$ is a dominating set of the UDG $G''=(V'',E'')$,
 where $V''=V'\setminus S$, $E''=E'\setminus \{(p_i,x_i),(x_i,y_i),(y_i,z_i) \mid 1\leq i \leq n\}$ and $|L'|=k+4l$. In order to ensure the liar's domination, every segment of
 length greater than or equal to 4 in $G'$
 should have at least two joint points in $L'$. If there are more than
 two joint points corresponding to a segment in $L'$, then we remove
 and/or replace the joint points so that each segment will have only
 two joint points while ensuring the domination. Now, $L'$ has been
 updated. Let $L''$ is the set obtained after updating $L'$ and
 $L''$ is also a dominating set of $G''$ with cardinality at most $k+2l$.
 
 We obtain the required dominating set $D$ of $G$
 from $L''$ as follows : consider a series of line segments, say
 $I=[p_i,p_j]$, corresponding an edge $(p_i,p_j)$ of $G''$, where $|I| = 4l'$ i.e.,
 $I$ has $l'$ segments. If none of $p_i$ and $p_j$ are in $L''$, then replace a point in
 $L''$ by $p_i$ with out loosing the domination property (existence
 of such a point is guaranteed as $L''$ is a dominating set). We apply this to all $I$'s.
 After applying the above process to all $I$'s, if there is an edge $(p_i,p_j)$ such that
 none of $p_i$ and $p_j$ are in $L'$, then there must exist $I_1=[p_s,p_i]$
 and $I_2=[p_t,p_j]$ with lengths $4l_1$ and $4l_2$ corresponding to
 some edges in $G''$ such that $p_s$ and $p_t$ are in $L''$. From the above preprocessing it
 is clear that $I_1$ and $I_2$ have at least $2l_1$ and
 $2l_2$ joint points in $L''$.

 From the above argument, there are at least $2l$ joint points in $L''$, where
 $l$ is the total number of line segments used in $G''$. This means that there are
 at most $|L''|-2l(=k)$ node points in $L''$.
 
 Let $D=\{v_i \in V \mid v_i \text{ corresponds to a node point in $L''$}\}$.
 So, $D$ is a dominating set of $G$ and $|D|\leq k$. Thus the sufficiency follows.

 Hence, UDG-LR-DOM is NP-complete.
\end{proof}

\section{Conclusion}
In this article we considered the liar's domination problem on unit disk graphs and
proved that the problem belongs to NP-complete class.

\end{document}